\documentclass{svproc}
\usepackage{url}

\usepackage[utf8]{inputenc}
\usepackage[T1]{fontenc}

\usepackage{amsmath,latexsym,amssymb,graphicx,dsfont,amsfonts}
\usepackage{algorithmic}
\usepackage{graphicx}
\usepackage{textcomp}
\usepackage{xcolor}

\begin{document}
\mainmatter              
\title{On Fairness in Voting Consensus Protocols}
\titlerunning{Fairness in Voting Consensus Protocols}  
%
\author{Sebastian M\"{u}ller\inst{1} \and Andreas Penzkofer\inst{2} \and
Darcy Camargo\inst{2} \and Olivia Saa\inst{2}}
\authorrunning{M\"uller, Penzkofer, Camargo, Saa} 
%
%
\institute{Aix Marseille Universit\'e, CNRS, Centrale Marseille, I2M - UMR 7373, 13453 Marseille, France,\\
\email{sebastian.muller@univ-amu.fr}
\and
IOTA Foundation, 10405 Berlin, Germany\\
\email{ \{andreas.penzkofer,  darcy.camargo, olivia.saa\}@iota.org}}

\maketitle              

\begin{abstract}
Voting algorithms have been widely used as consensus protocols in the realization of fault-tolerant systems. These algorithms are best suited for distributed systems of nodes with low computational power or heterogeneous networks, where different nodes may have different levels of reputation or weight.  Our main contribution is the construction of a \emph{fair} voting protocol in the sense that the influence of the eventual outcome of a given participant is linear in its weight.  Specifically, the fairness property guarantees that any node can actively participate in the consensus finding even with low resources or weight.
We investigate effects that may arise from weighted voting, such as loss of anonymity, centralization, scalability, and discuss their relevance to protocol design and implementation. 
\end{abstract}

\section{Introduction} 
\subsection{Preliminaries}
Weighted voting in distributed systems is known to increase efficiency 
and network reliability, 
but it also raises additional risks. 
It deviates from the \emph{one node - one vote}-principle  and  allows a defense against Sybil attacks if the weight corresponds to \emph{resources} that can be verified via resource testing or recurring costs and fee, e.g.,~\cite{Neil_asurvey}. However, weighted voting may induce a loss of  anonymity for the nodes and may incentivize centralization.

In most of the systems that are based on weighted voting, all participants are allowed to vote, and a centralized entity counts the votes and takes a decision. In the decentralized setting, participants or nodes have to find consensus on the outcome of the vote.  In the case where the number of nodes is high, and not all members can communicate with all other nodes, protocols where nodes only sample a certain number of nodes can be used to aggregate information, \cite{MoNeTa:14}, or find consensus, \cite{GaKuLe:78,MoDiAm:04,KaMo:07,cruise2013probabilistic,GoMaMaBe:15,fpc,fpcsim}.

This article focuses on a certain class of voting consensus protocols, namely, binary majority consensus. Some basic algorithms in this protocol class are simple majority consensus, \cite{MoNeTa:14}, Gacs-Kurdyumov-Levin \cite{GaKuLe:78}, and random neighbors majority, \cite{GoMaMaBe:15}.
The basic idea of these majority voting protocols is that nodes query other nodes about their current opinion, and adjust their own opinion throughout several rounds based on the proportion of other opinions they have observed. 

Despite these facts, they have been successfully applied not only in decision making but also in a wide range of engineering and economic applications \cite{banisch2010,niu2015,przyby2011},  and lead to the emerging science of sociophysics~\cite{castellano2009}.  
Recently, \cite{fpc} introduced the fast probabilistic consensus (FPC) protocol, an amelioration of the classical consensus voting protocols that is efficient and robust in Byzantine infrastructure. This protocol is studied in more detail in \cite{fpcsim} and serves this note as a reference for a voting consensus protocol.

\subsection{Weights and fairness}
The main contribution of this work is to define an adaption of the majority consensus protocol to a setting that allows nodes to have different weights. We define \emph{voting power}, Section \ref{sec:fairness}, as a power index to describe the influence of each node on the outcome of the consensus protocol. 
In comparison with standard voting models where the whole population is sampled, we have two instances where the weights may come in: firstly, in the sampling of the nodes to query and secondly, in the weighting of the votes to apply the majority rule. An essential consequence of fairness of the protocol is that it allows  defense against Sybil attacks.  In a decentralized and permissionless setting,  a malicious actor may gain a disproportionately large influence on the voting by creating a large number of pseudonymous identities that can increase its share of the voting power.

Moreover, fairness allows even nodes with very few resources or weight to participate in the consensus finding. This property is particularly important in networks where the sum of the weights of nodes with small weight is considerable. 

Besides technical and economic considerations, we want to mention the possible social impacts of an unfair system. For instance, unfair situations can make participants of the network unhappy, and this should be a real consideration in judging  the efficiency and adaption of a protocol, e.g., \cite{Rabin:91}. In particular, this is of importance in community-driven projects such as IOTA.

\subsection{IOTA}
IOTA is an open-source distributed ledger and cryptocurrency designed for the Internet of things. For the next generation of the IOTA protocol \cite{coordicide} introduces \textit{mana} in various places in order to obtain fairness. For instance, the protocol uses mana to obtain fairness in rate control, \cite{ViWeGaDi:19,vigneri2020}, where adaptive difficulty property guarantees that any node, even with low hashing power, can achieve similar throughput for given mana. This note discusses how mana is used as a weight in FPC to construct a fair consensus protocol. 


\subsection{Outline}
The rest of the paper organizes as follows. After giving an overview of previous work in Section \ref{sec:SoA}, we give a brief introduction to voting consensus protocol and FPC in \ref{sec:FPC}. The main contribution is the formulation of a proper mathematical framework in Section \ref{sec:fairness} and the construction of  a weighted voting consensus protocol that is \emph{fair} in the sense that the voting power is proportional to the weights of the nodes. Section \ref{sec:Zipf} proposes modeling of the weight distribution using a Zipf law. Under the validity of the Zipf law, we discuss in  Section \ref{sec:messageComplexity} impacts of the weighted voting on scalability and implementation of the protocol. Section \ref{sec:simresults} present some simulation results that show the behavior of the protocol in Byzantine infrastructure for different degrees of centralization of the weights. We conclude in Section \ref{sec:discussion} with a discussion.

\section{Related work}\label{sec:SoA}
\subsection{Weighted voting consensus protocols}
Voting consensus protocols (without weights) are widely studied in theory and applications, and they play an important role in social learning. Also, weighted voting systems have a long history in election procedures,  and often one is interested in measuring the influence of the power of the given participants, \cite{Ga:1994}.
Despite these facts, we were not able to find related results on voting consensus protocols with weights. Recently, \cite{MuMa:20} describes a model that is related to FPC, and that considers biased or stubborn agents. For the sake of brevity, we refer to  \cite{fpc,fpcsim} for more details on related work and references therein. 
\subsection{Fairness}
Fairness plays a prominent role in many areas of science and applications. It is, therefore, not astonishing that it plays its part also in DLT. For instance, PoW in Nakamoto consensus ensures that the probability of creating a new block is proportional to the computational power of a node; see \cite{ChPaCr:19} for an axiomatic approach to block rewards and further references. In PoS-based blockchains, the probability of creating a new block is usually precisely proportional to the node's balance. However, this does not always have to be the \emph{optimal} choice, \cite{PopovNxt,LeRePi:20}.

\section{Voting consensus protocols}\label{sec:FPC}
We give a brief definition of binary voting protocols in this section. More details can be found, for instance, in \cite{GaKuLe:78,MoDiAm:04,KaMo:07,cruise2013probabilistic,GoMaMaBe:15}. We refer to \cite{fpc,fpcsim} for more information on the fast probabilistic consensus (FPC).

To define the protocol accurately, we need some notation.
We assume the network to have $N$ nodes indexed by $1,2,\ldots, N$, and that every node can query any other nodes. We make this assumption for the sake of a better presentation; a node does not need to know every other node in the network. In fact, simulation studies \cite{GoMaMaBe:15,fpcsim} indicate that it is sufficient if every node knows about half of the other nodes. Moreover, it seems to be a reasonable assumption that nodes with high weights are known to every participant in the network.  Every node $i$ has an opinion or state. We note $s_{i}(t)$ for the opinion of the node $i$ at time $t$. Opinions take values in $\{0,1\}$. Every node $i$ has an initial opinion $s_{i}(0)$.

At each (discrete) time step each node chooses $k$ random nodes $C_{i}=C_i(t)$ and queries their opinions. Denote by $k_{i}(t)\leq k$  the number of replies received by node $i$ at time $t$ and set $s_j(t)=0$ if the reply from $j$ is not received in due time.  The updated mean opinion is then 
\begin{equation*}
\eta_{i}(t+1)=\frac1{k_{i}(t)} \sum_{j\in C_{i}} s_{j}(t).
\end{equation*}
Note that the neighbors $C_{i}$ of a node $i$ are chosen using sampling with replacement, and hence repetitions are possible.

As in \cite{fpcsim} we consider a basic version of the FPC introduced in \cite{fpc} in choosing some parameters by default. Specifically, we remove the cooling phase of FPC and the randomness of the initial threshold $\tau$. Let $U_{t}$, $t=1, 2,\ldots$ be i.i.d.~random variables with law $\mathrm{Unif}( [\beta, 1-\beta])$ for some parameter $\beta \in [0,1/2]$. The update rules for the opinion of a node $i$ is then given by
\begin{equation*}
s_{i}(1)=\left\{ \begin{array}{ll}
1, \mbox{ if } \eta_{i}(1) \geq \tau, \\
0, \mbox{ otherwise,}
\end{array}\right.. 
\end{equation*}

For subsequent rounds, i.e.\ $t\geq 1$:
\begin{equation*}
s_{i}(t+1)=\left\{ \begin{array}{ll}
1, \mbox{ if } \eta_{i}(t+1) > U_{t}, \\
0, \mbox{ if } \eta_{i}(t+1) < U_{t}, \\
s_{i}(t), \mbox{ otherwise.}
\end{array}\right. 
\end{equation*}
Note that if $\tau=\beta=0.5$, FPC reduces to a standard majority consensus. An asymmetric choice of $\tau$, $\tau\neq 0.5$ allows the protocol the distinction between the two kinds of integrity failure, \cite{fpc,fpcsim}. It is important, that the above sequence of random variables $U_t$ are the same for all nodes; we refer to \cite{fpcsim} for a more detailed discussion on the use of decentralized random number generators.

In contrast to many theoretical papers on majority dynamics, a local termination rule is needed for practical applications. Every node keeps a counter variable \verb?cnt? that is incremented by $1$ if there is no change in its opinion, and that is set to $0$ if there is a change of opinion. Once the counter reaches a certain threshold $\verb?l?$, i.e.,~$\verb?cnt?\geq \verb?l?$,  the node considers the current state as final. 
In the absence of autonomous termination the algorithm is halted after $\verb?maxIt?$ iterations. 
\vspace{0.9cm}
\section{Fairness}\label{sec:fairness}
In this section, we propose a proper mathematical framework of fairness. 
We consider a network of $N$ nodes whose  weight is described by $\{m_1,..,m_N\}$ with $\sum^N_{i=1} m_i=1$. In the sampling of the queries a  node $j$ is chosen now with probability  
\begin{equation}
p_j=\frac{f(m_j)}{\sum_{i=1}^N f(m_i)}.
\end{equation}
Each opinion is weighted by $g_j=g(m_j)$, resulting in the value
\begin{equation}
\eta_i(t+1)= \frac{1}{\sum_{j\in C_i} g_j} \sum_{j\in C_i} g_j s_j(t).
\end{equation}
The other parts of the protocol remain unchanged. 

We denote by $y_i$ the number of times a node $i$ is chosen. As the sampling is described by a multinomial distribution we can calculate the expected value of a query as
\vspace{-0.1cm}
\begin{equation}
\mathbb{E}\eta(t+1)=\sum_{i=1}^N s_i(t)v_i,
\end{equation}
\vspace{-0.1cm}
where
\begin{equation}
v_i=\sum_{\textbf{y}\in \mathbb{N}^N:\sum{y_i}=k}
\frac{k!}{y_1!\cdot\cdot\cdot y_N!}
\frac{y_i g_i }{\sum_{n=1}^N y_n g_n}
\prod^N_{j=1}p_j^{y_j}
\end{equation}
is called the \emph{voting power} of node $i$. The voting power measures the influence of the node $i$.  We would like the voting power to be proportional to the weight since this would induce a robustness to splitting and merging, \cite{LeSt:19}.
\begin{definition}[Robust to splitting  and Sybil attacks]
A voting scheme is \emph{robust to Sybil attacks} if a node $i$ splits into nodes $i_1$ and $i_2$ with a weight splitting ratio $x\in (0,1)$, then
\begin{equation}\label{eq:Sybil}
   v_i(m_i)\geq v_{i_1}(xm_i)+v_{i_2}((1-x)m_i).
\end{equation}
\end{definition}

\begin{definition}[Robust to merging]
A voting scheme is \emph{robust to merging} if a node $i$ splits into nodes $i_1$ and $i_2$ with a weight splitting ratio $x\in (0,1)$, then
\begin{equation}\label{eq:mere}
   v_i(m_i)\leq v_{i_1}(xm_i)+v_{i_2}((1-x)m_i).
\end{equation}
\end{definition}

\begin{definition}[Fairness]
A voting scheme $(f,g)$ is \emph{fair} if it is robust to Sybil attacks and robust to merging. In other words, if a node $i$ splits into nodes $i_1$ and $i_2$ with a weight splitting ratio $x\in (0,1)$, then
\begin{equation}\label{eq:fair}
   v_i(m_i)=v_{i_1}(xm_i)+v_{i_2}((1-x)m_i).
\end{equation}
\end{definition}
In the case where $g\equiv 1$, i.e.,~the $\eta$ is an unweighted mean, we show the existence of a voting scheme that is fair for all possible choices of $k$ and weight distributions.
\begin{theorem}\label{thm1}
For $g\equiv 1$ the voting scheme $(f,g)$ is fair if and only if $f$ is the identity function $f=id$.
\end{theorem}
\begin{proof}
We consider $g\equiv 1$. In this case we can simplify (\ref{eq:fair}) to 
$$v_i=\frac{1}{k}\sum_{\textbf{y}\in \mathbb{N}^N:\sum {y_i}=k}y_i\mathbb{P}[\textbf{y}],$$
where $\textbf{y}$ follows a multinomial distribution using the probability vector $\{p_j\}$ and making $k$ selections. Hence, we obtain
\begin{equation}
v_i(m_i)=p_i=\frac{f(m_i)}{\sum_{j=1}^N f(m_j)}.
\end{equation}
Now, let $S=\sum_{j=1}^N f(m_j)$ and 
\begin{equation}
\Delta_x = f(x m_i)+ f((1-x) m_i)- f(m_i).
\end{equation}
The fairness condition (\ref{eq:fair}) becomes
\begin{equation}
    \frac{f(m_i)}{S} = \frac{f(x m_i)}{S+\Delta_x} + \frac{f((1-x)m_i)}{S+\Delta_x},\quad \forall x\in(0,1),
\end{equation}
which is equivalent to
\begin{equation}
    \Delta_x f(m_i) = S \Delta_x,\quad \forall x\in(0,1).
\end{equation}
This means that either $f(m_i)=S$, meaning that there is only one node, or that $\Delta_x\equiv 0$, meaning that for all $x\in (0,1)$ we have that
\begin{equation}
    f(m)=f(xm)+f((1-x)m),
\end{equation}
and that $f$ is a linear function. 
\end{proof}

On the other hand, if the nodes are queried at random without weighting the probability for selection, i.e., $f\equiv 1$, then there exists no voting scheme that is fair for all $k$.
\begin{theorem}\label{thm2}
For $f\equiv1$ there exists no voting scheme $(f,g)$ with  for all $g_i>0$ that is fair for all $k$.
\end{theorem}
\begin{proof}
In the case of $f\equiv 1$ the voting power simplifies to
\begin{equation}
    v(m_i)=\frac{1}{N^k}\sum_{\textbf{y}\in \mathbb{N}^N:\sum{y_i}=k} \frac{k!}{y_1!\cdot\cdot\cdot y_N!}
\frac{ y_i g_i }{\sum_{n=1}^N y_n g_n}.\\
\end{equation}
We will consider a special situation that does not satisfy the fairness condition. We consider the situation with $N=2$ nodes and $m_1=2/3$ and $m_2=1/3.$ The voting power of the first node is then
\begin{equation}\label{eq:thm2}
    v\left(\frac23\right)=\mathbb{E} \left[ \frac{X g(\frac23)}{X g(\frac23) + (k-X) g(\frac13)} \right],
\end{equation}
where $X$ follows a binomial distribution $\mathcal{B}(k,1/2).$
Now, it suffices to show that the right hand side of Equation (\ref{eq:thm2}) is not constant in $k$. Due to the law of large numbers and the dominated convergence theorem we have that
\begin{equation}
    \mathbb{E} \left[ \frac{X g(\frac23)}{X g(\frac23) + (k-X) g(\frac13)} \right] \stackrel{k \to \infty}{\longrightarrow} \frac{ g(\frac23)}{g(\frac23)+ g(\frac13)}.
\end{equation}
We conclude by observing  that
\begin{equation}
    \mathbb{E} \left[ \frac{X g(\frac23)}{X g(\frac23) + (k-X) g(\frac13)} \right] < \frac{ g(\frac23)}{g(\frac23)+ g(\frac13)},
\end{equation}
where the strict inequality follows from Jensen's Inequality using that $X$ is not a constant almost surely and $g(\frac23)>g(\frac13).$ 
\end{proof}
For these reasons we fix from now on $g\equiv 1$ and $f=id$.

\section{Distribution of weight}\label{sec:Zipf}
The existence of a fair voting scheme is independent of the actual distribution of weight. However, to make a more detailed prediction of the qualities of the voting consensus protocol, it may be appropriate to make assumptions on the weights. 

Probably the most appropriate modelings of the weight distributions rely on universality phenomena. 
The most famous  example of this universality phenomenon is the central limit theorem. While the central limit is suited to describe statistics where values are of the same order of magnitude, it is not appropriate to model more heterogeneous situations where the values might differ in several orders of magnitude. For instance,  values in (crypto-)currency systems are not distributed equally; \cite{btcdistribution}.  

Zipf's law and the closely related Pareto distribution describe mathematically various models of real-world systems. For instance, many economic models, \cite{wealth_pareto}, use these laws to describe the wealth and influence of participants.  These laws govern the asymptotic distribution of many statistics which
\begin{enumerate}
    \item take values as positive numbers;
    \item range over many different orders of magnitude;
    \item arise from a complicated combination of largely independent factors; and
    \item have not been artificially rounded, truncated, or otherwise constrained in size.
\end{enumerate}
The Zipf law is the appropriate variant for modeling the weight distribution and is defined as follows: The $n$th largest value $y(n)$ 
should obey an approximate power law, i.e.,~it should be approximately 
\begin{equation}
    y(n)=C n^{-s}
\end{equation}
 for the first few $n=1,2,3,\ldots$ and some parameters $s > 0$ and normalising constant $C$. We refer to \cite{Zipf} for an excellent introduction to this topic. 


A convenient way to observe a Zipf law is by plotting the data on a log-log graph, with the axes being log(rank order) and log(value). The data conforms to a Zipf law to the extent that the plot is linear, and the value of $s$ can be found using linear regression. 
For instance, Fig. \ref{fig:IOTAzipf} shows the distribution of IOTA for the top 100 richest addresses with a fitted Zipf law.

Due to the universality phenomenon, the plausibility of hypotheses 1) - 4) above, and Fig. \ref{fig:IOTAzipf}, we assume, in the following sections, a Zipf law for the weight distribution. 
\begin{figure}
    \center
    \includegraphics[width=0.8\textwidth,trim={0 0 0 0cm},clip]{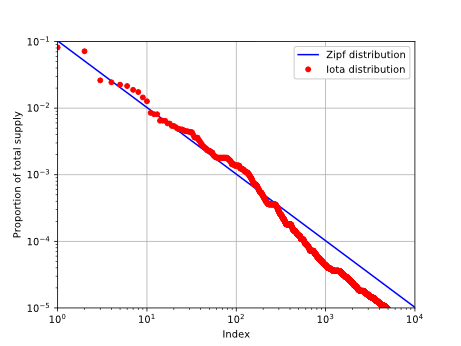}
    \caption{\small Distribution of relative IOTA value on the top $10.000$ addresses 
    with a fitted Zipf law with $s=1$; as of July $2020$.}
    \label{fig:IOTAzipf}
\end{figure}

\section{Scalability and Message complexity}\label{sec:messageComplexity}
An essential property of voting consensus protocols is their scalability. In fact, at every round, every node is queried on average $k$ times indifferent to the network size. In our proposed fair voting schemes where nodes are sampled proportional to their weight, this is no longer true, and nodes with higher weight are queried more often. This affects the scalability of the protocol and might generate feedback on the weight distribution.

\begin{lemma}\label{lem:complex1}
We assume the weights to follow a Zipf law with parameters $(s,N)$. Then, the average number of queries a node  of rank $h(N)$ receives per round  is of order (as $N\to\infty)$ 
\begin{enumerate}
    \item $\Theta(N^s h(N)^{-s})$, if $s<1$;
    \item $\Theta(\frac{N}{\log{N}}h(N)^{-1})$, if $s=1$;
    \item $\Theta(N h(N)^{-s})$, if $s=1$.
\end{enumerate}
\end{lemma}
\begin{proof}
At every round a node of rank $h(N)$ queried on average 
\begin{equation}
    N\cdot  \frac{h(N)^{-s}}{\sum_{n=1}^N n^{-s}}
\end{equation} times.
If $s< 1$ this becomes asymptotically $\Theta(N^s h(N)^{-s})$, if $s=1$ we obtain $\Theta(\frac{N}{\log{N}}h(N)^{-1})$, and if $s>1$ this is  $\Theta(N h(N)^{-s})$. 
\end{proof}

In particular, the node with the highest weight, i.e.,~$h(N)=1$, is queried $\Theta(N^s), \Theta(\frac{N}{\log{N}})$, or $\Theta(N)$ times, and might eventually be overrun by queries. Nodes whose rank is $\Theta(N)$ have to answer only $\Theta(1)$ queries. This is in contrast to the case $s=0$ where every node has the same weight, and every node is queried on average a constant number of times.

 Nodes with high weights are therefore incentivized to communicate their opinions through different communication channels, e.g.,\ to gossip their opinions on an underlying peering network and not to answer each query separately. Since not all nodes can gossip their opinions (in this case, every node would have to send  $\Omega(N)$ messages), we have to find a threshold when nodes gossip their opinions or not. 

If we assume that nodes with high weights have higher throughput than nodes with lower weights a reasonable threshold is $\log(N)$, i.e.,~only the $\Theta(\log(N))$ highest weights nodes do gossip their opinions, leading to $\Theta(\log{N})$ messages for each node in the gossip layer. In this case the expected number of queries the highest weight node, that is not allowed to gossip its opinions,  receives is  $\Theta((\frac{N}{\log{N}})^{s})$
if $s< 1$, $\Theta(\frac{N}{(\log{N})^2})$ if $s=1$, and $\Theta(\frac{N} {(\log{N})^{s}})$ if $s>1$.
In this case, nodes of rank between $\Theta(\log{N})$ and $\Theta(N)$ are the critical nodes for message complexity.  

A possible consequence might be that these \emph{middle rank} nodes might gossip their opinion leading to possible congestion of the network or might stop to answer queries leading to less security of the consensus protocol. A less selfish response of such a node could be to either split up the node into several nodes or to pool with other nodes to gossip their opinions according to the above threshold rule. This, however, might influence the distribution of weight. 

The above situation may apply well to heterogeneous networks, i.e., the computational power and bandwidths of the node might differ in several orders of magnitudes, \cite{ViWeGaDi:19,vigneri2020}.  In homogeneous networks the above issues may be solved by a \emph{fair} attribution of message complexity. 

\begin{lemma}
We assume the weights to follow a Zipf law with parameters $s$ and $N$. Then there exists a \emph{fair} threshold for gossiping opinions such that every node has to process the same order of messages. The message complexity for each node is $O(\sqrt{N})$ for all choices of $s$.
\end{lemma}
\begin{proof}
We have to find a threshold such that the maximal number of queries a node has to answer should equal the number of messages that are gossiped. For $s<1$ this leads to the following equation
\begin{equation}
N^s h(N)^{-s} = h(N)    
\end{equation}
and hence we obtain that a threshold of order $N^{\frac{s}{s+1}}$ leads to $\Theta(N^{\frac{s}{s+1}})$ messages for every node to send. For $s>1$ one obtains similarly a threshold of $N^{\frac1{1+s}}$ leading to $\Theta(N^{\frac1{1+s}})$ messages. In the worst case, i.e.,~$s=1$, the message complexity for each node in the network is $O(\sqrt{N})$. 
\end{proof}


\section{Simulations}\label{sec:simresults}
In this section, we present a short simulation study that shows how the performance of the FPC depends on the distribution of the weights. 

\subsection{Threat model}
We assume that an adversary holds a proportion $q$ of the total weight and that it splits its weight  equally between its $q N$ nodes such that each node holds $1/N$ of the total weight. We assume that the adversary is at every moment aware of all opinions and transmits at time $t+1$ the opinion of the mana-weighted minority of the honest nodes of step $t$. 
\subsection{Failures}
In the case of heterogeneous weight distributions, there are different possibilities to generalize the standard failures of consensus protocols: namely integration failure, agreement failure, and termination failure. In this note,  we consider only agreement failure since, in the DLT context, this failure turns out to be the most severe. In the strictest sense, an agreement failure occurs if not all nodes decide on the same opinion. We consider the $1\%$-agreement failure; such a failure occurs if at least $1\%$ of the nodes differ in their final decision.

\subsection{Results}

We choose the initial average opinion $p_0$ equal to the value of the first threshold $\tau$. This can be considered as the critical case since for values of $p_0\gg\tau$ or $p_0\ll \tau$, the agreement failure rate is so small that numerical simulation is no longer feasible.
The initial opinion is assigned as follows. The highest mana nodes that hold together more than $p_0$ of the mana are assigned opinion $1$ and the remaining opinion $0$.
Other default parameters are: $N=1000$ nodes, $p_0=\tau=0.66, \beta=0.3, \verb?l?=10, \verb?maxIt?=50$.

In Fig.\ \ref{fig:q} we investigate the protocol with a relatively small quorum size, $k=20$, and study  the agreement failure rate as a function of $q$.
\begin{figure}[ht]
\centering
     \includegraphics[width=0.8\textwidth,trim={0 0 0 0cm},clip]{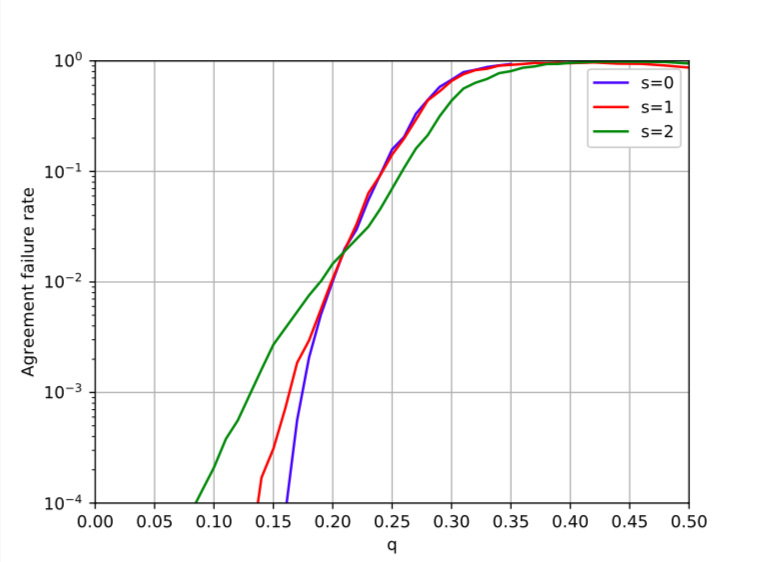}
  \caption{Agreement failure rates with $q$ for three different weight distributions and $k=20$ \newline}
  \label{fig:q}
\end{figure}
We see that the performance  depends on the parameter $s$; the higher the centralization is the lower the agreement failures are for high $q$, but at the price that the protocols performs worse if $q$ is smaller. In Fig. \ref{fig:k} we  observe an exponential decay of the agreement failure rate in the quorum size $k$.  

\begin{figure}[ht]
\centering
     \includegraphics[width=0.8\textwidth,trim={0 0 0 0cm},clip]{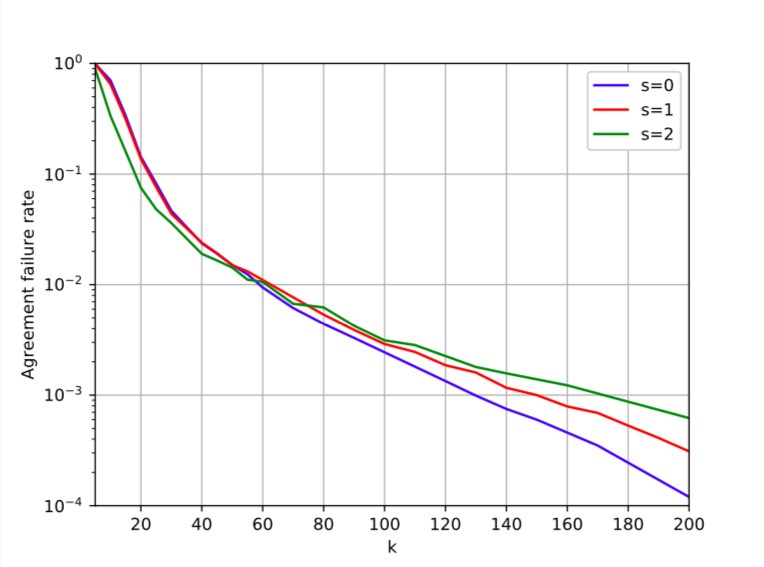}
  \caption{ Agreement failure rates with $k$; $q=0.25$.  \newline}
  \label{fig:k}
\end{figure}
The source code of the simulations is made open source and available online.\footnote{The link is not given in this version, since it might reveal the identity of the authors.}

\section{Discussion}\label{sec:discussion}

We proposed a mathematical framework to study fairness in consensus voting protocols and constructed a fair voting scheme, Section \ref{sec:fairness}. Even though this voting scheme is robust to splitting and merging, there are \emph{second order} effects that may incentivize nodes to optimize their weights. One example we studied concerns the message complexity and its consequences, Section \ref{sec:messageComplexity}. Other secondary effects that may influence the weights are, for instance; basic resource costs such as maintaining nodes, network redundancy, and service availability. Moreover, the fact that the security of the protocol has an impact on its performance may lead to changes in the weight distribution.

Nodes with high weight are likely no longer be anonymous. In the case where these nodes gossip their opinions, and other users widely accept their reputation, some nodes may decide not longer to take part in the consensus finding but to only follow the nodes with the highest reputations. This would give disproportional weight to nodes with high weight, leading to an unfair situation.

A complete mathematical treatment of the above is not realistic; all the more the evolution of such a permissionless and decentralized consensus protocol depends also on other components of the network, as well as economic and psychological elements. 
\bibliographystyle{ieeetr}
\bibliography{bibliography}

\end{document}